\documentclass[letterpaper,10pt]{article}

\usepackage[svgnames]{xcolor}

\usepackage{amsmath, amsfonts, amssymb, amsthm}
\usepackage{enumitem}

\usepackage[nottoc]{tocbibind} %

\usepackage{float}
\usepackage[labelfont=bf,textfont=it]{caption}

\newtheorem{Thm}{Theorem}[section]
\newtheorem{Cor}[Thm]{Corollary}
\newtheorem{Lem}[Thm]{Lemma}

\numberwithin{equation}{section}

\newcommand{\R}{{\mathsf{R}}}

\newcommand{\f}{g}
\newcommand{\lambdadef}{\max\left(21, \left\lceil \tfrac{40}{3}(T-1)\ln D\right\rceil\right)}
\newcommand{\ldef}{2\lceil \log_\lambda D \rceil}
\newcommand{\gammadef}{ \lceil \max(8\log_\lambda D, 8\ln\tfrac{1}{\mu}) \rceil}
\newcommand{\mdef}{ \left\lceil\log\tfrac{1}{2\mu} + 2\log T + 2\log(1+\ell/4)\right\rceil }
\newcommand{\mdefsmall}{ \lceil\log(1/(2\mu)) + 2\log(T) + 2\log(1+\ell/4)\rceil}
\newcommand{\sdef}{ \lceil \log_q(2D+1) \rceil }

\newcommand{\coeff}{\mathtt{coeff}}

\newcommand{\expons}{\mathtt{exponents}}
\newcommand{\COL}{\mathcal{C}}

\newcommand{\FF}{{\mathbb{F}}}
\newcommand{\RR}{{\mathbb{R}}}
\newcommand{\bigoh}{\mathcal{O}}

\newcommand{\softoh}{\widetilde{\mathcal{O}}}
\newcommand{\softO}{\widetilde{\mathcal{O}}}  %

\newcommand{\ssum}{\textstyle\sum}
\newcommand{\SLP}{{\mathcal{S}}}
\newcommand{\tuples}{\mathtt{tuples}}
\newcommand{\ZZ}{{\mathbb{Z}}}
\newcommand{\Fq}{{\mathbb{F}_q}}

\DeclareMathOperator{\llog}{loglog}

\usepackage[section]{placeins}

\usepackage[round,longnamesfirst]{natbib}

\usepackage[vlined,ruled,boxed]{algorithm2e}
\SetKw{br}{break}

\definecolor{darkgreen}{rgb}{0,.35,0}
\definecolor{darkblue}{rgb}{0,0,.35}
\definecolor{darkred}{rgb}{.35,0,0} \usepackage[pdfauthor={Andrew
  Arnold, Mark Giesbrecht and Daniel S. Roche}, pdftitle={Sparse interpolation over finite fields via low-order
  roots of unity}, colorlinks,linkcolor=DarkBlue,
citecolor=DarkGreen,urlcolor=darkred, bookmarksnumbered]{hyperref}

\begin{document}

\title{Sparse interpolation over finite fields 
via low-order roots of unity}

\author{
Andrew Arnold\\
\small Cheriton School of Computer Science\\
\small University of Waterloo\\
\href{https://cs.uwaterloo.ca/~a4arnold/}{\tt a4arnold@uwaterloo.ca}
\and
Mark Giesbrecht\\
\small Cheriton School of Computer Science\\
\small University of Waterloo\\
\href{https://cs.uwaterloo.ca/~mwg/}{\tt mwg@uwaterloo.ca}
\and
Daniel S.\ Roche\\
\small Computer Science Department\\
\small United States Naval Academy\\
\href{http://www.usna.edu/cs/roche/}{\tt roche@usna.edu}
}

\maketitle

\begin{abstract}
  We present a new Monte Carlo algorithm for the interpolation of a
  straight-line program as a sparse polynomial $f$ over an
  arbitrary finite field of size $q$.  We assume {\em a priori} bounds
  $D$ and $T$ are given on the degree and number of terms of $f$.  The
  approach presented in this paper is a hybrid of the diversified and
  recursive interpolation algorithms, the two previous fastest known
  probabilistic methods for this problem. By making effective use of
  the information contained in the coefficients themselves, this new
  algorithm improves on the bit complexity of previous methods by a
  ``soft-Oh'' factor of $T$, $\log D$, or~$\log q$.
\end{abstract}

\section{Introduction}

Let $\FF_q$ be a finite field of size $q$ and consider a ``sparse'' polynomial
\begin{equation}
  \label{eqn:f}
  f  = \ssum_{i=1}^t c_iz^{e_i} \in \FF_q[z],
\end{equation}
where the $c_1,\ldots,c_t\in\FF_q$ are nonzero and the exponents
$e_1,\ldots,e_t\in\ZZ_{\geq 0}$ are distinct.  Suppose $f$ is provided
as a straight-line program, a simple branch-free program which
evaluates $f$ at any point (see below for a formal definition).
Suppose also that we are given an upper bound $D$ on the degree of $f$
and an upper bound $T$ on the number of non-zero terms $t$ of $f$.
Our goal is to recover the standard form \eqref{eqn:f} for $f$, that
is, to recover the coefficients $c_i$ and their corresponding
exponents $e_i$ as in \eqref{eqn:f}.  Our main result is as follows.

\begin{Thm}\label{thm:cost}
  Let $f \in \FF_q[z]$ with at most $T$ non-zero terms and degree at most $D$,
  and let $0 < \epsilon \le 1/2$.  Suppose we are given a
  division-free straight-line program $\SLP_f$ of length $L$ that
  computes $f$.  Then there exists an algorithm (presented below) that
  interpolates $f$, with probability at least $1 - \epsilon$, with a
  cost of
  \begin{equation*}
    \softoh\left( L T \log^2 D
      \left( \log D + \log q \right)
      \log\tfrac{1}{\epsilon}
    \right)%
    \footnote[2]{For
      functions $\phi, \psi : \RR_{>0} \rightarrow \RR_{>0}$, we
      say $\phi \in \softoh(\psi)$ if and only if $\phi \in \bigoh(\psi(\log
      \psi)^c)$ for a constant $c \geq 0$.}
  \end{equation*}
  bit operations.
\end{Thm}
This cost improves on previous methods by a factor of $T$, $\log D$,
or $\log q$, and may lay the groundwork for even further
improvements.  See Table \ref{tab:comp} below for a detailed
comparison of the complexity of various approaches to this problem.

\subsection{Straight-line programs and sparse polynomials}

The interpolation algorithm presented in this paper is for
straight-line programs, though it could be adapted to other more
traditional models of interpolation.  Informally, a straight-line
program is a very simple program, with no branches or loops,
which evaluates a polynomial at any point, possibly in an extension
ring or field. Straight-line programs serve as a very useful model to
capture features of the complexity of algebraic problems (see, e.g.,
\citep{Str90}) especially with respect to evaluation in extension rings, as
well as having considerable efficacy in practice (see, e.g.,
\citep{FIKY88}).

More formally, a division-free \emph{Straight-Line Program} over a ring
$\R$, henceforth abbreviated as an \emph{SLP}, is a branchless sequence of
arithmetic instructions that represents a polynomial function.  It
takes as input a vector $(a_1, \dots, a_K)$ and outputs a vector
$(b_1, \dots, b_L)$ by way of a series of instructions $\Gamma_i : 1
\leq i \leq L$ of the form $\Gamma_i : b_i \longleftarrow \alpha \star
\beta$, where $\star$ is an operation $'+', '-'$, or $'\times'$, and
$\alpha, \beta \in \R \cup \{a_1, \dots, a_K\} \cup \{b_0, \dots,
b_{i-1}\}$.  The inputs and outputs may belong to $\R$ or a ring
extension of $\R$.  We say a straight-line program \emph{computes} $f
\in \R[x_1, \dots, x_K]$ if it sets $b_L$ to $f(a_1, \dots, a_K)$.

The straight-line programs in this paper compute over finite fields
$\Fq$ with $q$ elements, and ring extensions of $\Fq$.  We  assume
that elements of $\Fq$ are stored in some reasonable representation
with $O(\log q)$ bits, and that each field operation requires
$\softO(\log q)$ bit operations.  Similarly, we assume that elements
in a field extension $\FF_{q^s}$ of $\Fq$ can be represented with
$O(s\log q)$ bits, and that operations in $\FF_{q^s}$ require
$\softO(s\log q)$ bit operations.  

Each of the algorithms described here determines $f$ by {\em probing}
its SLP: executing it at an input of our choosing and observing the
output.  This is analogous to evaluation of the polynomial at a point
in an extension of the ground field in the more traditional
interpolation model.  To fairly account for the cost of such a probe
we define the \emph{probe degree} as the degree of the ring extension
over $\Fq$ in which the probe lies.  A probe of degree $u$ costs
$\softoh(Lu)$ field operations, or $\softO(Lu\log q)$ bit operations.
The \emph{total probe size} is the sum of the probe degrees of all the
probes used in a computation.

Polynomials are also stored with respect to the power basis (powers of
$x$) in a {\em sparse representation}.  In particular, $f$ as in
\eqref{eqn:f} would be stored as a list $[(c_i,e_i)]_{1 \leq i \leq
  t}$.  Given such a polynomial $f$, we let $\coeff(f, k)$, denote the
coefficient of the term of $f$ of degree $k$ (which may well be zero).
We let $\expons(f)$ denote the sorted list of exponents of the nonzero
terms in $f$, and $\#f$ denote the number of nonzero terms of $f$.
When we write $r(z) = f(z) \bmod d(z)$, we assume that the modular
image $r$ is reduced, i.e., it is stored as the remainder of $f$
divided by $d$, and $\deg r < \deg d$. We frequently refer to such $r$
as ``images'' of the original polynomial $f$, as they reveal some
limited amount of information about $f$.

The problem of (sparse) interpolation can then be seen as one of
conversion between representations: Efficiently transform a polynomial
given by a straight-line program into a ``standard'' sparse
representation with respect to a power basis in $x$ as in \eqref{eqn:f},
given ``size'' constraints $D$ and $T$ as~above.

\subsection{Previous results}

The following table gives a comparison of existing algorithms for the
sparse interpolation of straight-line programs.

\begin{table}[h]
  \setlength{\tabcolsep}{5pt}
  \centering\renewcommand{\arraystretch}{1.2}
  \begin{tabular}{|l| r| c|}
    \hline
    & \multicolumn{1}{c|}{Bit complexity} & Type \\
    \hline
    Dense & $LD\log q$ & Det \\\hline
    Garg \& Schost & $LT^4\log^2 D\log q$ & Det \\\hline
    LV G \& S & $LT^3\log^2 D\log q$ & LV \\\hline
    Diversified & $LT^2\log^2 D(\log D+\log q)$ & LV \\\hline
    Recursive & $LT\log^3 D\log q \cdot \log\frac{1}{\epsilon}$ & MC \\\hline
    This paper & $LT\log^2 D(\log D + \log q) \cdot \log\frac{1}{\epsilon}$ &
    MC \\\hline
  \end{tabular}
  \caption{
    A comparison of interpolation algorithms for
    straight-line programs (ignoring polylogarithmic factors).\\ 
    Det=Deterministic, LV=Las Vegas, MC=Monte Carlo.
  \label{tab:comp}}
\end{table}

Most of these are probabilistic algorithms.  By a \emph{Las Vegas}
algorithm we mean one that runs in the expected time stated but
produces an output of guaranteed correctness.  A \emph{Monte Carlo}
algorithm takes an additional parameter $\epsilon\in (0,1)$ and
produces an output guaranteed to be correct with probability at least
$1-\epsilon$.

The algorithm of \cite{GarSch09} finds a {\em good prime},
that is, a prime that separates all the terms of $f$.  For $f$ given
by \eqref{eqn:f},
\begin{equation}\label{eqn:goodimage}
  f \bmod (z^p-1) = \ssum_{i=1}^t c_iz^{e_i \bmod p},
\end{equation}
and so the terms of the {\em good image} $f \bmod (z^p-1)$ remain
distinct provided the exponents $e_i$ are distinct modulo $p$.  This
good image gives us $t=\#f$, which in turn makes it easy
to identify other good primes. Their algorithm then constructs the
symmetric polynomial $\chi(y) = \prod_{i=1}^t (y-e_i)$ by Chinese
remaindering of images $\chi(y) \bmod p_j$, for sufficiently many good
primes $p_j$, $1 \leq j \leq \ell$.  Note that the image $f \bmod (z^{p_j}-1)$ gives the values $e_i
\bmod p_j$ and hence $\chi(y) \bmod p_j$.  Given $\chi$, the algorithm
then factors $\chi$ to obtain the exponents $e_i$.  The corresponding coefficients 
of $f$ may be obtained by inspection of any good image \eqref{eqn:goodimage}.
Their algorithm can be made
faster, albeit Monte Carlo, by using randomness; we probabilistically
search for a good prime by selecting primes at random over a specified
range, choosing as our good prime $p$ one for which the image $f \bmod
(z^p-1)$ has maximally many terms.

An information-theoretic lower bound on the total probe size required
is $\Omega(T(\log D+\log q))$ bits, the number of bits used to encode
$f$ in \eqref{eqn:f}.  This bound is met by Prony's
\citeyearpar{Pro95} original algorithm, which requires a total probe
size of $O(T\log q)$ under the implicit assumption that $q>D$.  
Much of the complexity of the sparse interpolation problem appears to
arise from the requirement to accommodate \emph{any} finite
field. Prony's algorithm is dominated by the cost of discrete
logarithms in $\FF_q$, for which no polynomial time algorithm is known
in general.  When there is a choice of fields (say, as might naturally
arise in a modular scheme for interpolating integer or rational
polynomials) more efficient interpolation methods have been developed.
\cite{Kal88:frag} demonstrates a method for sparse interpolation over
$\FF_p$ for primes $p$ such that $p-1$ is smooth; see \citep{Kal10:pasco} for further
exposition.  In our notation, this algorithm would require an
essentially optimal $\softO(LT(\log D+\log q))$ bit
operations. Parallel algorithms and implementations for this case of
chosen characteristic are also given by \cite{JavMon10}.
Moreover, our need for a bound $T$ on the number of non-zero terms is
motivated by the hope of a Las Vegas or deterministic algorithm.  The
early termination approach of \cite{KalLee03} identifies $t$ with high
probability at no asymptotic extra cost.

\subsubsection{Diversified Interpolation}

We say a polynomial $g(z)$ is {\em diverse} if its non-zero terms have
distinct coefficients.  The interpolation algorithm of
\cite{GieRoc11} distinguishes between images of distinct non-zero terms by {\em
  diversifying} $f$, that is, choosing an $\alpha$ such that $f(\alpha
z)$ is diverse. This Monte Carlo algorithm entails three probabilistic
steps.  First, it determines $t=\#f$ by searching for a probable good
prime.  Second, it finds an $\alpha$ that diversifies $f \bmod
(z^p-1)$.  Finally, it looks for more good primes $p$, and constructs the
exponents $e_i$ of $f$ by way of Chinese remaindering on the
congruences $e_i \bmod p$.  This gives $f(\alpha z)$, from which it is
straightforward to recover $f(z)$.

The diversified interpolation algorithm was initially described for a
``remainder black box'' representation of $f$, a weaker model than a
straight-line program.  To search for an appropriate $\alpha \in
\FF_q$, the algorithm requires that the field size $q$ is greater than
$T(T-1)D$.
Under the SLP model we can easily adapt the algorithm to work when $q
\leq T(T-1)D$; simply choose $\alpha$ from a sufficiently large field
extension.  This slight adaptation increases the cost of a probe by a
factor of $\softoh((\log D + \log T)/\log q)$.

The cost of sparse SLP interpolation using the diversified
algorithm is then
$\softoh(\log D + \log\tfrac{1}{\epsilon})$ probes of degree 
$\softoh((T^2\log D)\lceil(\log D)/(\log q)\rceil)$, 
where $\epsilon\in (0,1)$ is a given bound on the probability of failure.  
Since each operation in $\FF_q$ costs $\softoh(\log q)$ bit
operations, the cost in bit operations~is

\begin{equation}
  \label{eqn:divcost}
  \softoh\left(
    LT^2\log D
    (\log D + \log \tfrac{1}{\epsilon})
    (\log D + \log q)
  \right).
\end{equation}
For fixed $\epsilon$ this cost becomes 
$\softoh( LT^2\log^2 D(\log D+\log q))$.

The diversified and probabilistic Garg-Schost interpolation algorithms
may be made Las Vegas (i.e., guaranteed error free) by way of a
deterministic polynomial identity testing algorithm. The fastest-known
method for certifying that polynomials (given by algebraic circuits)
are equal over an arbitrary field, by \cite{BlaHar09}, requires
$\softoh(LT^2\log^2 D)$ field operations in our notation.  The gains
in faster interpolation algorithms would be dominated by the cost of
this certification; hence they are fast Monte Carlo algorithms whose
unchecked output may be incorrect with controllably small probability.

\subsubsection{Recursive interpolation}

The algorithm of \cite{ArnGieRoc13} is faster than
diversified interpolation when $T$ asymptotically dominates either
$\log D$ or $\log q$.  The chief novelty behind that algorithm is to
use smaller primes $p$ with relaxed requirements.  Instead of
searching for good primes separating all of the non-zero terms of $f$, we
probabilistically search for an {\em ok prime} $p$ which separates
{\em most} of the non-zero terms of $f$.  Given this prime $p$ we then
construct images of the form $f \bmod (z^{pq_j}-1)$, for a set of
coprime moduli $\{q_1, \dots, q_\ell\}$ whose product exceeds $D$, in
order to build those non-colliding non-zero terms of $f$.

The resulting polynomial $f^*$ contains these terms, plus possibly a
small number of {\em deceptive terms} not occurring in $f$, such that
$f-f^*$ now has at most $T/2$ non-zero terms.  The algorithm then updates the
bound $T \leftarrow T/2$ and recursively interpolates the
difference $g=f-f^*$.

The total cost of the recursive interpolation algorithm is 
$\softoh(\log D + \log\tfrac{1}{\epsilon})$ probes of degree 
$\softoh( T\log^2 D)$, 
for a total cost of
\begin{equation}
  \label{eqn:reccost} 
  \softoh\left(
    LT\log^2 D
    \left(\log D + \log\tfrac{1}{\epsilon}\right)
    \log q
  \right)
\end{equation}
bit operations, 
which is $\softoh(L T \log^3 D \log q)$ when $\epsilon\in (0,1)$ is a constant.

\subsection{Outline of new algorithm and this paper}

As in \citep{ArnGieRoc13}, our new algorithm recursively builds an
approximation $f^*$, initially zero, to the input polynomial $f$ given
by an SLP.  The algorithm interpolates $g=f-f^*$ with bounds $D \geq
\deg(g)$ and $T \geq \#g$.  We update $T$ as we update $g$.

We begin with the details of our straight-line program model in 
Section~\ref{sec:slp}.
Section~\ref{sec:primes} describes how we choose a set of $\ell
\in \softoh(\log D)$ ok primes $p_i \in \softoh(T\log D)$, $1 \leq i
\leq\ell$.   Given these primes $p_i$, we then compute images $g_{ij} =
g(\alpha_j z) \bmod (z^{p_i}-1)$ for choices of $\alpha_j$ that will
(probably) allow us to identify images of like terms of $g$.  This
approach relies on a more general notion of diversification.  We explain
how we choose the values $\alpha_j$ and give a probabilistic analysis in
Section \ref{sec:div}.  Section \ref{sec:terms} details how we use
information from the images $g_{ij}$ to construct at least
half of the terms of $g$.

Section \ref{sec:recurse} describes how the algorithm iteratively builds $f^*$,
followed by a probabilistic and cost analysis of the algorithm as a whole
in Section~\ref{sec:cost}.
Some conclusions are drawn in Section~\ref{sec:conclusion}.

\section{Probing an SLP}
\label{sec:slp}

In this paper we will only consider a single-input, division-free SLP
$\SLP_f$ that computes a univariate polynomial $f$ over a finite
field.  We can reduce multivariate $f \in \FF_q[x_1, \dots, x_n]$ of
total degree less than $D$, by way of a {\em Kronecker
  substitution}. Interpolating the univariate polynomial
$f(z,z^D, \dots, z^{D^{n-1}})$ preserves the number of non-zero terms of $f$ and
allows the easy recovery of the original multivariate terms, but
increases the degree bound to $D^n$.   See \citep{ArnRoc14} for
recent advances on this topic.

We will {\em probe} $\SLP_f$, that is, execute the straight-line
program on inputs of our choosing.  One could naively input an
indeterminant $z$ into $\SLP_f$, and expand each resulting polynomial
$b_i$.  This sets $b_L$ to $f(z)$.  Unfortunately, intermediate
polynomials $b_i$ may have a number of non-zero terms, and be of degrees,
which grow exponentially in terms of $L$, $T$, and $D$.

Instead, we limit the size of intermediate results by selecting a
symbolic $p$-th root of unity $z \in \FF_q[z]/\langle z^p-1 \rangle$
as input, for different choices of $p$.  At each instruction we expand
$b_i \in \FF_q[z]$, and reduce modulo $(z^p-1)$.  This sets $b_L$ to
$f(z) \bmod (z^p-1)$.  Using the method of \cite{CanKal91}, the cost
of one such instruction becomes $\softoh(p)$ field operations.  More
generally we will produce images of the form $f(\alpha z) \bmod
(z^p-1)$, where $\alpha$ is algebraic with degree $s$ over $\FF_q$.
That is, we choose as input $\alpha z$, where $z$ is again a $p$-th
root of unity.  Here the instruction cost becomes $\softoh(p)$
operations in $\FF_{q^s}$, or $\softoh(ps)$ operations in $\FF_q$.  We
call $ps$ the {\em probe degree} of a probe with input $\alpha z$.

We will also have to produce such images of polynomials given
explicitly by sparse representations rather than implicitly in a
SLP.  Given a single term $cz^e$,
$e<D$, we can produce $c(\alpha z)^e \bmod (z^p-1)$ by computing
$\alpha^e$ via a square-and-multiply approach and reducing the
exponent of $z^e$ modulo $p$.
For $p<D$ the former step dominates this cost and requires $\softoh(s\log D)$ operations in $\FF_q$.
For $f^*$, a sum of at most $T$ such terms and $p < D$, the cost of constructing $f^*(\alpha z) \bmod (z^p-1)$ becomes $\softoh(sT\log D)$.

In particular, given an SLP $\SLP_f$ and a sparse polynomial $f^*$,
we will need to compute $f(\alpha z) - f^*(\alpha z) \bmod (z^p-1)$,
as described above.
Procedure~\ref{proc:ComputeImage} is a subroutine to perform such a
computation, and its cost is
$\softoh(Lsp + sT\log D)$ operations in $\FF_q$. When $p$ is at least of magnitude
$\Omega((T\log D)/L)$, the SLP probe dominates the cost, and it becomes
simply $\softoh(Lsp\log q)$ bit operations.

\begin{procedure}[ht]
\caption{ComputeImage($\SLP_f, f^*, \alpha, p$)}\label{proc:ComputeImage}
\KwIn{
	$\SLP_f$, an SLP computing $f\in\FF_q[z]$;
        $f^* \in \FF_q[x]$ given explicitly;
        $\alpha \in \FF_{q^s}$;
        integer $p\ge 1$.
}
\KwResult{$f(\alpha z) - f^*(\alpha z) \bmod (z^p - 1)$}

\smallskip

$g_{p,\alpha} \longleftarrow f(\alpha z) \bmod (z^p-1)$, 
  by probing $\SLP_f$ at $\alpha z$ over
  \phantom{$g_{p,\alpha} \longleftarrow$ }$\FF_{q^s}[z]/\langle z^p-1 \rangle$; \\

\ForEach{$e \in \expons(f^*)$}{
  $g_{p,\alpha} \longleftarrow g_{p,\alpha} 
    - \coeff(f^*,e) \cdot \alpha^e \cdot z^{e \bmod p}$
}

\Return{ $g_{p,\alpha}$ }
\end{procedure}

\section{Selecting primes}
\label{sec:primes}

The aim of Sections \ref{sec:primes}--\ref{sec:terms} 
is to build a polynomial $f^{**}$ that
contains at least half of the terms of $g=f-f^*$.  For notational
simplicity, in these sections we write $g=\ssum_{i=1}^t c_iz^{e_i}$
and let $T$ and $D$ bound $\#g$ and $\deg(g)$ respectively.

We will use images of the form $\f \bmod (z^p-1)$, $p$ a prime, in order
to extract information about the terms of $\f$.  We say two terms
$cz^{e}$ and $c'z^{e'}$ of $\f$ {\em collide} modulo $(z^p-1)$ if $p$
divides $e - e'$, that is, if their reduced sum modulo $(z^p-1)$ is a single term.  If a term $cz^{e}$ of $g$ does not collide
with any other term of $\f$, then its image $cz^{e \bmod p}$ contained
in $\f \bmod (z^p-1)$ gives us its coefficient $c$ and the image of the
exponent $e \bmod p$.

We let $\COL_\f(p)$ denote the number of terms of $\f$ that collide with
any other term
modulo $(z^p-1)$. We need $\COL_\f(p)$ to be small so that the next
approximation $f^*$ of $f$ contains many terms of $f$, but the primes
$p$ must also be small so that the cost of computing each $f^*$ is not too
great. In this technical section, we show how to bound the size of the
primes $p$ to balance these competing concerns and minimize the total
cost of the algorithm.

In the first phase of the algorithm, we will look for a set of primes $p_i \in \softoh(T\log D)$, and corresponding images
$$
\f_i = \f \bmod (z^{p_i}-1),\qquad  1 \leq i \leq \ell,
$$
such that, in each of the images $g_i$, most of the terms of $g$ are not
in any collisions.  More specifically, we
want the primes $(p_i)_{1 \leq i \leq \ell}$ and images $(g_i)_{1 \leq i
\leq \ell}$ to meet the following criteria:
\begin{enumerate}[label=(\roman{enumi})]
\item\label{item:half_terms}  At least half of the terms of $g$ do not collide modulo $(z^{p_i}-1)$ with any other term of $g$ for at least $\lceil \ell/2 \rceil$ of the primes $p_i$;
\item\label{item:half_col}  Any pair of terms of $g$ collide modulo $(z^{p_i}-1)$ for fewer than $\lceil \ell/2 \rceil$ of the primes $p_i$;
\item\label{item:prod_D} Any $\lceil \ell/2 \rceil$ of the primes $p_i$ have a product exceeding~$D$.
\end{enumerate}
Every nonzero term of an image $g_i$ is the image of a sum of terms of $g$.  
If we are able to collect terms of the images $g_i$, $1 \leq i \leq \ell$, that
are images of the same sum of terms of $g$, then from
\ref{item:half_col}, any such collection containing a term for at least
$\ell/2$ images $g_i$ must be an image of a single term of $g$.  By
\ref{item:half_terms}, at least half of the terms of $g$ will
produce such a collection.  By \ref{item:prod_D}, such
collections will contain sufficient information to reconstruct the exponent of the
corresponding term of $g$.  Thus, given a means of collecting terms in this way,
we will be able to construct half of the terms of $g$.

We note, since any two terms of $g$ have degree differing by at most
$D$, that \ref{item:prod_D} implies \ref{item:half_col}.  To satisfy
\ref{item:half_terms}, it suffices that $\COL_\f(p_i) \leq T/2$ for
each prime $p_i$.  We will accordingly call $p$ an {\em ok prime} in
this paper if $\COL_\f(p)<T/2$.  To this end we establish a range in
which most primes $p$ have $\COL_\f(p)<T/4$, half the desired bound.

\begin{Lem}[\citealt{ArnGieRoc13}]
\label{lem:lambda}
Let $\f \in \FF_q[z]$ be a polynomial with $t \leq T$ terms and degree
$d \leq D$.  Let
\begin{equation}\label{eqn:lambda}
\lambda = \lambdadef.
\end{equation}
Then fewer than half of the primes $p$ in the range $[\lambda, 2\lambda]$ satisfy $\COL_g(p) \geq T/4$.
\end{Lem}
We will look for ok primes in the range $[\lambda, 2\lambda]$.  To satisfy \ref{item:prod_D}, we will select $\ell = \ldef$
primes.  Corollary~\ref{Cor:lambda_ell} gives us a means of selecting a set of primes for which a constant proportion of those primes $p$ have fewer than $T/4$ colliding terms.
\begin{Cor}
\label{Cor:lambda_ell}
Suppose $\f, D, T$, and $\lambda$ are as in Lemma \ref{lem:lambda}.
Let $0 < \mu < 1$, and suppose
\begin{align}\label{eqn:ell_lambda}
\ell &= \ldef ,& \gamma &= \gammadef.
\end{align}
If we randomly select $\gamma$ distinct primes from $[\lambda,
2\lambda]$, then $\COL_\f(p) \leq T/4$ for at least $\ell$ of the
chosen primes $p$, with probability at least $1-\mu$.
\end{Cor}

\begin{proof}
By Lemma \ref{lem:lambda}, at least half of the primes $p$ in $[\lambda,
2\lambda]$ have $\COL_\f(p)<T/4$.  We require that a proportion of $\tfrac{1}{4}$ of the $\gamma$ primes randomly selected to have this property.

In order to obtain a bound on at least $\gamma/4$ of the primes having this
property, define random variables $X_i$, for $1 \le i \le \gamma$, to be 1 if
the $i$th prime $p_i$ satisfies $\COL_\f(p_i)<T/4$, and $0$ otherwise. From
the reasoning above, we know that $E[X_i] \ge \tfrac{1}{2}$. Therefore, the
expected number of primes with this property, $E[\Sigma X_i]$, is at least
$\tfrac{\gamma}{2}$.

Hoeffding's inequality provides proof that the \emph{actual} value of $\Sigma X_i$
is not too much smaller than this expected value. Specifically, Theorems 1 and 4
of \citep{Hoe63} show that
\[
\Pr\{\Sigma X_i \le \tfrac{\gamma}{4}\} \le 
    \exp\left(-2\left(E[\Sigma X_i] - \tfrac{\gamma}{4}\right)^2\right)
    \le \exp\left(\tfrac{-\gamma}{8}\right).
\]
For $\gamma \ge
8\ln\tfrac{1}{\mu}$, this is less than~$\mu$. 
\end{proof}

We will generate some $\gamma$ primes, of which $\ell$ primes $p$ have $\COL_\f(p)<T/4$.  In order to identify some primes $p$ for which $\COL_\f(p)$ is low, we have the following lemma.

\begin{Lem}[\cite{ArnGieRoc13}, Corollary 11]
  \label{Lem:choose_p}
  Suppose $\f \bmod (z^q-1)$ has $s_q$ non-zero terms, and $\f \bmod
 (z^p-1)$ has $s_p$ non-zero terms, with $s_p \geq s_q$.  Then $\COL_\f(p) \leq
  2\COL_\f(q)$.
\end{Lem}

Corollary \ref{Cor:lambda_ell} guarantees with high
probability that at least $\ell$ primes $p_i$ of the $\gamma$ selected
satisfy $\COL_\f(p_i) < T/4$. Call these the ``best primes''.
Unfortunately, there is no easy way to determine $\COL_\f(p_i)$ at this point,
so we do not know which $\ell$ primes are the best primes!

Our solution to this seeming dilemma is to order the primes in decreasing
order of $\#\f_i$. Now consider any of the first $\ell$ primes $p_i$ in this
ordering. Since $\#\f_i$ is among the largest $\ell$ values, we know that
$\#\f_i \ge \#\f_j$ for some $p_j$ that is one of the ``best primes'' and
satisfies $\COL_\f(p_j) < T/4$. Then by Lemma~\ref{Lem:choose_p},
we know that $\COL_\f(p_i) \le 2\COL_\f(p_j)$, which is less than $T/2$.

In other words, the first $\ell$ primes in this ordering are not necessarily
the ``best'', but they are good enough because they all satisfy
$\COL_\f(p_i) < T/2$.
This method is
described in procedure \ref{proc:FindPrimes} below.

The statement of Corollary \ref{Cor:lambda_ell} assumes 
that there are a certain number of primes in $[\lambda, 2\lambda]$, in
order to pick $\gamma$ of them.
Letting $n$ be the actual number of primes in this range, we
require that $n \geq \gamma=\gammadef$.  This puts further constraints
on $\lambda$.  By Corollary 3 of \cite{RosSch62}, and using the
definition of $\lambda$ from Lemma~\ref{lem:lambda},
\begin{equation*}
n \geq 3\lambda/(5\ln\lambda) \geq 8(T-1)\log_\lambda D,
\end{equation*}
which is at least $8\log_\lambda D$ for $T \geq 2$.  Thus, for $T>1$,
we only require that $n \geq \lceil 8\ln\tfrac{1}{\mu}\rceil$.  
If the number of primes $n$ is even smaller than this, one could simply
compute all $n$ primes in the interval and use them all instead of picking a
random subset.
Since at least half these primes have $\COL_\f(p)<T/4$, we would only
require that $n\geq 2\ell$, which must be true since
$n\ge 8\log_\lambda D$.

To ensure that computing all $n$ primes in the interval does not
increase the overall cost of the algorithm, consider that in this case
$\gamma$ exceeds $3\lambda/(5\ln\lambda)$.
Then the upper bound of Corollary 3 of \cite{RosSch62} gives
\begin{equation*}
n \leq \lceil 7\lambda/(5\ln\lambda)\rceil \in O(\gamma).
\end{equation*}
Therefore whether we choose only $\gamma$ primes from the interval or
are forced to compute all $n$ of them, the number of primes used is
always $O(\gamma)$.

We make one further note on the cost of producing $\gamma$ primes at
random.  In a practical implementation, one would likely choose numbers
at random (perhaps in a manner that avoids multiples of small primes),
and perform Monte Carlo primality testing to verify whether such a
number is prime.  Problems could arise if the algorithm produced
pseudoprimes $p_i$ and $p_j$ that are not coprime.  Thus one would
also have to consider the failure probability of primality testing in
the analysis of such an approach.

For the purposes of our analysis, we generate all the primes up to
$2\lambda$ with $\softoh(\lambda) = \softoh(T\log D)$ bit operations
using a wheel sieve \citep{Pri82}, which does not dominate the cost in
\eqref{eqn:cost2}, and guarantees that all of the chosen $p$ are
actually prime.

The total cost is $\gamma$ probes of degree at most $2\lambda$, where
\begin{align*}
\lambda &\in \softoh(T\log D), &\text{and}&&
\gamma &\in \softoh\left(\frac{\log D}{\log T} + \log\tfrac{1}{\mu}\right),
\end{align*}
 for a field-operation cost of 
\begin{equation}\label{eqn:cost2}
\softoh(L\gamma\lambda) = \softoh\left( L\left(\log D + \log\tfrac{1}{\mu}\right)T\log D \right).
\end{equation}

\begin{procedure}[ht]
\caption{FindPrimes($\SLP_f, f^*, T, D, \mu$)}\label{proc:FindPrimes}
\KwIn{
	$\SLP_f$, an SLP computing $f\in\FF_q[z]$;
        $f^* \in \FF_q[x]$ given explicitly;
        $T \geq \max(\#g,2)$; $D \geq \deg(g)$; $0 < \mu < 1/3$;
        where $g$ is defined as the unknown polynomial $f-f^*$.
}

\KwResult{A list of primes $p_i$ and 
        images $g(z) \bmod (z^{p_i}-1)$ 
        such that $\prod_{i=1}^\ell p_i \ge D$ and, 
        with probability exceeding $1-\mu$,
        $\COL_g(p_i) < T/2$ for $1 \leq i \leq \ell$.
}\smallskip

$\lambda \longleftarrow \lambdadef$\\
$\gamma \longleftarrow \gammadef$\\
$\ell = \ldef$\\

\If{$\gamma \leq 3\lambda/(5\ln\lambda)$}{
	$\mathcal{P} \longleftarrow$ \text{ $\gamma$ primes chosen at random from $[\lambda, 2\lambda]$ }\\
}\Else{
$\mathcal{P} \longleftarrow$ \text{ all primes from $[\lambda, 2\lambda]$ }\\
}

\lForEach{$p \in \mathcal{P}$}{
  $g_p \longleftarrow$ \ref{proc:ComputeImage}$(\SLP_f,f^*,1,p)$
}
$(p_1, p_2, \dots ) \longleftarrow \mathcal{P} \text{ sorted in decreasing order of $\#g_{p_i}$}$\\
\Return{ $(p_1, \dots, p_\ell), (g_{p_1}, \dots, g_{p_\ell})$ }
\end{procedure}

\section{Detecting deception}
\label{sec:div}

At this stage we have probably found primes $p_i$,
$\COL_\f(p_i) \leq T/2$, and their corresponding images $g_i$, for $1
\leq i \leq \ell$.  The challenge now remains to collect terms amongst
the images $\f_i$ that are images of the same term of $\f$.  For our
purposes we need a more general notion of diversification than that
introduced by \cite{GieRoc11}.  To this end we will construct images
\begin{equation}\label{eqn:gij}
\f_{ij} = \f(\alpha_j z) \bmod (z^{p_i}-1),
\end{equation}
where $\alpha_j \neq 0$ belongs to $\FF_q$ or a field extension $\FF_{q^s}$.

Any term of an image $g_i$ is an image of either a single term of $\f$, or a
sum of multiple terms of $\f$. Our algorithm needs to identify and
discard the terms in $g_i$ corresponding to multiple terms of $g$, using
only the single terms to reconstruct the actual terms in $g$. To analyse
this situation, consider the bivariate polynomials
\[
g(yz) \bmod (z^{p_i}-1) = h_{i,0} + h_{i,1} z + \cdots + h_{i,{p_i-1}}
z^{p_i-1},
\]
where each $h_{i,u} \in \FF_{q^s}[y]$ is the sum of the terms in $g$
with degrees congruent to $u$ modulo $p_i$. Each $h_{i,u}$ has
between 0 and $T$ terms and degree at most $D$ in $y$. For a given
$\alpha$, computing $g_i(\alpha z)$ gives the univariate polynomial
whose coefficient of degree $u$ is $h_{i,u}(\alpha)$.

Consider a single term in the unknown polynomial $g$. If that term does
not collide with any others mod $p_i$, then for some $u$, $h_{i,u}$
consists of that single term. If the same term does not collide modulo
$p_k$, then there exists some $v$ such that $h_{i,u} = h_{k,v}$. Obviously,
for any $\alpha\in\FF_{q^s}$, we will have $h_{i,u}(\alpha) =
h_{k,v}(\alpha)$, and our algorithm can use this correlation to
reconstruct the term, since its exponent equals $u$ mod $p_i$ and $v$
mod $p_k$. Based on the previous section, most of the terms in $g$ will
not collide modulo most of the $p_i$'s, and so there will be sufficient
information here to reconstruct those terms.

The problem is that we may have
$h_{i,u}(\alpha) = h_{k,v}(\alpha)$, but $h_{i,u}\ne h_{k,v}$. We call
this a \emph{deception}, since it may fool our algorithm into
reconstructing a single term in $g$ that does not actually exist. Our
algorithm will evaluate with multiple choices $\alpha_j$, $1 \leq j \leq m$, for $\alpha$, and we ``hope''
that, whenever $h_{i,u}\ne h_{k,v}$, at least one of  $\alpha_j$'s gives
$h_{i,u}(\alpha_j) = h_{k,v}(\alpha_j)$. In this case we say $\alpha_j$
\emph{detects} the deception. The following lemma provides this hope.

\begin{Lem}\label{Lem:detect}
Let $\f \in \FF_q[z]$ be a polynomial of degree at most $D$ and at most $T$ nonzero terms.
Let $p_i$ be a prime such that $\COL_\f(p_i) < T/2$ and let $g_i=g \bmod (z^{p_i}-1)$ for $1 \leq i \leq \ell$.
Let
\begin{align*}
m&= \mdef, \text{ and }\\
s&=\sdef.
\end{align*}
Choose $\alpha_1, \dots, \alpha_m$ at random
from $\FF_{q^s}^*$.  Then, with probability at least $1-\mu$, every deception
amongst the images $g_1, \dots, g_\ell$ is detected by at least one of the
$\alpha_j$.
\end{Lem}

\begin{proof}
Consider $h_{i,u}, h_{k,v}$ as above, with $h_{i,u} \neq h_{k,v}$.  As $(h_{i,u}-h_{k,v})(y)$ has degree at most $D$, there are at most $D$ choices of $\alpha$ for which $h_{i,u}(\alpha) = h_{k,v}(\alpha)$.

Thus, if $q^s > 2D$ and $\deg(h_{i,u}-h_{k,v})\leq D$, then at most half of
the $\alpha \in \FF_{q^s}^*$ can comprise a root of a $h_{i,u}-h_{k,v}$.
Thus, if we then select $\alpha_1, \dots, \alpha_m$ at random from
$\FF_{q^s}^*$ and construct images $g(\alpha_jz ) \bmod (z^{p_i}-1)$
for $i=1,\dots,\ell$, $j=1,\dots,m$, we will fail to detect a single given
deception with probability at most
\begin{equation}
  \label{eqn:proofdetect}
  \begin{aligned}
  \left(\frac{1}{2}\right)^m 
  & =  \left(\frac{1}{2}\right)^{\mdefsmall}\\
  & \leq \frac{\mu}{\tfrac{1}{2}T^2(1+\ell/4)^2}~.
  \end{aligned}
\end{equation}
As $\f$ has at most $T$ terms, there are at most $2^T-1$ possible
choices for $h_{i,u}\neq 0$; however, only a small proportion of these
choices may correspond to a term in an image $g_i$.  Each $g_i$, $1
\leq i \leq \ell$, contains at most $\COL_\f(p_i)/2 \leq T/4$ terms that are images of a
sum of at least two terms of $\f$.  The remaining nonzero terms of $g_i$ are images of
single terms of $\f$, of which there are at most $T$.
Thus there are
fewer than $T(1+\ell/4)$ distinct sums corresponding to one of the nonzero $h_{i,u}$.

A deception occurs between a pair of such sums, thus there are fewer than $\tfrac{1}{2}T^2(1+\ell/4)^2$ deceptions.  It then follows from \eqref{eqn:proofdetect} that the probability that at least one deception is not detected by any of $\alpha_1, \dots, \alpha_m$ is bounded above by $\mu$.
\end{proof}

\enlargethispage{6pt}

It is important to note that $m$ is logarithmic in $T$ and $\ell$, so
that a multiplicative factor of $m$ will not affect the ``soft-Oh'' cost of
the algorithm in terms of $T$ or $D$.  If $q \leq 2D$, then we need
instead to work in an extension of $\FF_q$ of degree
$s=\lceil\log_q(2D+1)\rceil$.  The cost of computing $g_{ij}$, $1 \leq i
\leq\ell, 1 \leq j \leq m$, in terms of $\FF_q$-operations, becomes
\begin{equation*}
\softoh( L\ell m \lambda s ) = \softoh\left(LT\log^2 D\left\lceil \frac{\log D}{\log q}\right\rceil \log\tfrac{1}{\mu}\right),
\end{equation*}
which dominates the cost \eqref{eqn:cost2} of Section \ref{sec:primes}.

\section{Identifying images of like terms}
\label{sec:terms}

As we construct the images $g_{ij}$, we will build vectors of
coefficients of images $g_{ij}$.  Namely, for every congruence class
$e \bmod p_i$ for which there exists a nonzero term of degree $e$ in
at least one image $g_{ij}$, we will construct a vector $v^{i,e} \in
\FF_{q^s}^m$, where $v^{i,e}_j$ contains the coefficient (possibly
zero) of the term of $g_{ij}$ of degree $e$.

We will use the vectors $v^{i,e}$ to identify terms of the images $g_i$
that are images of like terms of $g$.  We use these vectors $v^{i,e}$ as
keys in a dictionary $\tuples$.  
Each value in the dictionary is comprised of 
a list of those tuples $(e,i)$ for which $v^{i,e}=v$. 

Provided the probabilistic steps of Sections \ref{sec:primes} and
\ref{sec:div} succeeded, if a key $v$ is found more than $\ell/2$ times,
then it corresponds to a single, distinct term of $g$, as opposed to a
sum of terms of $g$. This is indicated by the size of the list
$\tuples(v)$ being at least~$\ell/2$.

The dictionary $\tuples$ should be an ordered dictionary that supports
logarithmic-time insertion and retrieval. Any balanced binary search tree,
such as a red-black tree will be suitable (see, e.g., \citep{CLRS01}, Chapter 13).
To set $\tuples(v)$, we first search to see if $v$ is an existing key;
if not, an empty list is first inserted as the value of $\tuples(v)$.

A red-black tree
of size $n$ requires $\mathcal{O}(\log n)$ comparisons for insert and
search operations.  We compare keys $v \in \FF_{q^s}^m$ lexicographically,
which may entail $m$ comparisons of elements in $\FF_{q^s}$.  Each
comparison therefore requires $\bigoh(ms\log q)$ bit operations.

The number of different vectors $v^{i,d}$ that will appear is bounded
above by the number of distinct subsets of terms of $g$ which can
collide modulo $(z^{p_i}-1)$, for any of the primes $p_i$. Since there
are $\ell$ primes, and at most $T$ terms in $g$, there are no more
than $T\ell$ vectors which will be inserted as keys into $\tuples$.
Thus, the cost of constructing this tree is
\begin{equation}
  \label{eqn:cost4}
  \begin{aligned}
   &\bigoh(T\ell ms \cdot \log q \cdot \log(T\ell)) \\
   & ~~ = \softoh\left( T\log D
    \left(\log D + \log q\right) \log\tfrac{1}{\mu} \right)
  \end{aligned}
\end{equation}
bit operations.  This cost is dominated by that of constructing
the $g_{ij}$.  Each term in each $g_{ij}$ contains an element of
$\FF_{q^s}$ and an exponent at most $2\lambda$.  This requires $\softoh(
s\log q + \log \gamma)$ bits, which is
$\softoh(\log D + \log q + \log T + \log \tfrac{1}{\mu})$.  
The additional bit-cost of traversing the
images $g_{ij}$ and appending to the lists in $\tuples(v)$ is reflected in the
cost of their construction in Section \ref{sec:div}.

After we have constructed the dictionary $\tuples$,
we traverse it again to build terms
of $g$.  For every key $v$ whose corresponding list has size at least
$\ell/2$, we have all
the pairs $(i,d)$ such that $v^{i,d} = v$.  What remains is to
construct the term corresponding to the key $v$.  We reconstruct the
exponent by Chinese remaindering on the first $\ell/2$ congruences $d
\bmod p_i$.  As each exponent is at most $D$, the cost of constructing
one exponent is bounded by $\softoh( \log^2 D )$ bit operations.  Thus the
total bit-cost of Chinese remaindering becomes $\softoh( T\log^2 D)$.
As the bit cost of $s$ operations in $\FF_q$ is $\softoh(\log q + \log
D)$, this cost of Chinese remaindering is bounded asymptotically by
$\eqref{eqn:cost4}$, the cost of constructing $\tuples$ itself.

We obtain the coefficient by inspection of $g_i$.  We sum all of these constructed terms into a polynomial $f^{**}$ approximating $g=f-f^*$.

Procedure \ref{proc:BuildApproximation} restates the method described to
construct $f^{**}$.  
For the sake of brevity, this procedure does not detail the data
structures used for the polynomials $g_{ij}$ and $f^{**}$ that it
constructs. In order to achieve the stated complexity bounds, these
sparse polynomials must be implemented by dictionaries mapping exponents
to coefficients that support logarithmic-time insertion and retrieval.
As with $\tuples$, a red-black tree, hash table, or similar standard
data structure could be used. Converting between such a representation
and the usual list of coefficient-exponent pairs is trivial.

\begin{procedure}[ht]
\caption{BuildApproximation($\SLP_f, f^*, T, D, \mu$)}\label{proc:BuildApproximation}
\KwIn{
  $\SLP_f$, an SLP computing $f \in \FF_q[z]$;
  $f^* \in \FF_q[x]$ given explicitly;
  $T \geq \#g$; $D \geq \deg(g)$; 
  $0 < \mu < 1/3$;
  where $g$ is defined as the unknown polynomial $f-f^*$.
}

\KwResult{
  $f^{**}$ such that $g-f^{**}$ has at most $T/2$ terms, with probability greater than $1-\mu$.
}
\bigskip

$(p_1,\ldots,p_\ell), (g_1,\ldots,g_\ell)
  \longleftarrow$ \ref{proc:FindPrimes}$(\SLP_f, f^*, T, D, \mu)$\\

$m \longleftarrow \mdef$\\
$s \longleftarrow \sdef$\\
$(\alpha_1, \dots, \alpha_m) \longleftarrow$ $m$ randomly chosen nonzero elements
\phantom{$(\alpha_1, \dots, \alpha_m) \longleftarrow$} from $\FF_{q^s}$;\\
$\tuples \longleftarrow $ dictionary mapping $\FF_{q^s}^m$ to lists of pairs of
\phantom{$\tuples \longleftarrow $} integers;\\
\medskip

\For{$i \longleftarrow 1$ \KwTo $\ell$}{

	\ForEach{$e \in \bigcup_{1 \le j \le m} \expons(g_{ij})$}{ 
    \smallskip

    $v \longleftarrow \big(\coeff(g_{i1},e), 
      \ldots, \coeff(g_{im},e)\big)$ \\
    $\tuples(v) \longleftarrow \tuples(v),\ (i,e)$ \\
  }
}\medskip

$f^{**} \longleftarrow 0$ \\

\ForEach{$v \in \tuples$}{
  \If{$\#\tuples(v) \ge \ell/2$}{
		$e \longleftarrow$ solution to congruences
		\phantom{$e \longleftarrow$}$\{ e_i \bmod p_i \mid (i,e) \in \tuples(v)\}$;\\
		$c \longleftarrow \coeff(g_i, e)$, for any of the
		\phantom{$c \longleftarrow$}$(i, e) \in \tuples(v)$;\\
		$f^{**} \longleftarrow f^{**} + cz^e$
  }
}\medskip

\Return{$f^{**}$}
\end{procedure}

\section{Updating our approximation}\label{sec:recurse}

Recall we have a polynomial $f^*$ approximating $f$ given by our straight-line
program.  We construct a polynomial $f^{**}$ that is comprised of at least
$T/2$ terms of $g=f-f^*$.  Once we have $f^{**}$, we set $f^* \leftarrow f^* +
f^{**}$, $T \leftarrow \lfloor T/2 \rfloor$, and
repeat the process until $T$ is $1$, at which point $g$ consists of (at most)
a single nonzero term.

We thus execute this process at most $\log T$ times, where $T$ is the
initial bound on the number of non-zero terms of $f$.  Recall that the steps of sections
\ref{sec:primes} and \ref{sec:terms} each succeed with probability greater than
$1-\mu$.  Thus, if we would like the algorithm to succeed with probability
greater than $1-\epsilon > 1/2$, we can set 
\begin{equation}\label{eqn:mu}
\mu = \epsilon/(2\log T).
\end{equation}
As $T\ge 2$ and $\epsilon \le 1/2$, we 
will always have $\mu \le 1/4$, which satisfies the constraint $\mu<1/3$
from Section~\ref{sec:primes}.

When $T=1$, $g$ is a single term and its coefficient is given by $g(1)$.  The
exponent of the term comprising $g$ may be computed from $g \bmod (z^p-1)$ for the first $\log D$ primes $p$.  This cost is $\log D$ probes of degree
$\softoh(\log D)$, or a cost of $\softoh(\log^2 D \log q)$ bit operations.

We give the interpolation algorithm in procedure 
\ref{proc:MajorityVoteSparseInterpolate}.  We call the algorithm 
``majority-vote'' sparse interpolation, as we effectively require a
majority of the images $f_i$, $1 \leq i \leq \ell$, to vote on whether
a sum of terms of $f$ is in fact a single term.

\begin{procedure}[ht]
\caption{MajorityVoteSparseInterpolate($\SLP_f, T, D,\epsilon$)}\label{proc:MajorityVoteSparseInterpolate}

\KwIn{$\SLP_f$, an SLP computing $f \in \FF_q[z]$; $T \geq \# f$; $D \geq \deg(f)$; $0 < \epsilon < 1/2$.}
\KwResult{The sparse representation of $f$, with probability at least $1-\epsilon$.}
\smallskip

$(f^*, \mu) \longleftarrow (0,\ \epsilon/(2\log T))$\\

\While{$T > 1$}{

	$f^{**} \longleftarrow$ \ref{proc:BuildApproximation}$(\SLP_f, f^*, T, D, \mu/2)$\\
	$f^* \longleftarrow f^* + f^{**}$\\
	$T \longleftarrow \lfloor T/2 \rfloor$\\
}\bigskip
$(c,e) \longleftarrow ($\ref{proc:ComputeImage}$(\SLP_f,f^*,1,1),0)$\\
\lIf{$c = 0$}{ \Return $f^*$ }
\For{the first $\log D$ primes $p$ }{
$e' \longleftarrow$ degree of the single term in \\
\hspace{24pt}  \ref{proc:ComputeImage}$(\SLP_f,f^*,1,p)$;\\
Update $e$ with $e'$ modulo $p$ by Chinese remaindering.
}

\Return $f^* + cz^e$
\end{procedure}

\section{Cost analysis}\label{sec:cost}

We now are ready to analyze the cost of the algorithm
\ref{proc:MajorityVoteSparseInterpolate} and verify
Theorem~\ref{thm:cost}.  As we have argued in Sections
\ref{sec:primes}--\ref{sec:terms}, the cost of one iteration of the
algorithm is dominated by the cost of constructing the images
$g_{ij}$.  Recall that there are $\ell$ primes $p_i$, each less than
or equal to $2\lambda$, and there are $m$ elements $\alpha_j$'s, each
in an extension $\FF_{q^s}$. The total number of field operations is
thus $\softoh(L\ell m\lambda s)$. The values of $L, T, D$, and
$\epsilon$ are specified in the input. Recall the following
parameters from equation \eqref{eqn:mu},
Corollary~\ref{Cor:lambda_ell} and Lemma~\ref{Lem:detect}:

\begin{align*}
  \mu &= \epsilon/(2\log T) 
  \Rightarrow \log\tfrac{1}{\mu} \in \softoh(\log\tfrac{1}{\epsilon} + \llog T), \\
  \lambda &= \lambdadef \in \softoh(T\log D), \\
  \ell &= \ldef \in \softoh( (\log D)/(\log T)),\\
  m &= \mdef, \\
  &~~~~~~\in \softoh(\log\tfrac{1}{\epsilon} + \log T + \llog D), \\
  s &= \sdef \in \softoh(1 + (\log D)/(\log q)).
\end{align*}
Therefore the total cost of the $\softoh(L\ell m\lambda s)$ operations
in $\FF_q$ is
\begin{equation}
  \softoh\left(
    L T \log^2 D \left(\log D + \log q\right)\log\tfrac{1}{\epsilon}
  \right)
\end{equation}
bit operations.  The multiplicative $\log T$ factor due to the number
of iterations does not affect the ``soft-Oh'' cost above.  This cost
is a multiplicative factor of $\softoh(T)$ improvement over the cost
\eqref{eqn:divcost} of diversified interpolation \citep{GieRoc11}.  It
also improves over the cost \eqref{eqn:reccost} of recursive
interpolation \citep{ArnGieRoc13} by a multiplicative factor of
$\softoh(\min(\log D, \log q))$.

\section{Conclusions}
\label{sec:conclusion}

We have presented a new algorithm for sparse interpolation over an
arbitrary finite field that is asymptotically faster than those
previously known. In terms of bit operations it improves on previous
methods by ``soft-Oh'' multiplicative factor of $T$, $\log
D$, or $\log q$.

In this ``majority-vote'' interpolation, we combine the main ideas of the
previous two algorithms. Namely, we reduce the probe degree by only
aiming to reconstruct {\em some} of the terms of $f$ at every iteration;
and we distinguish images of distinct terms (and subsets of terms) by
way of diversification.

We mention a few open and motivating problems. First, we believe the
algorithm in this paper has considerable potential in the more
traditional numerical (floating point) domain.  There, an unknown
sparse polynomial is reconstructed from a small number of evaluations
in $\mathbb{C}$ under a standard backward-error model of precision and
stability.  See \citep{GieRoc11} for an example of a straight-line
program interpolation algorithm adapted to floating point
computation. Our hope is that by evaluating at lower-order roots of
unity (such as used in this paper) we can provably increase the
numerical stability over the Prony-like algorithm of \cite{GLL09},
while maintaining its near-optimal efficiency.

Second, the algorithm presented here is Monte Carlo. It remains
unknown whether there exist faster deterministic or Las Vegas
polynomial identity tests that may render recursive or majority-vote
interpolation Las Vegas of the same complexity.

Finally, as noted earlier, an information theoretic lower-bound on
sparse interpolation suggests a minimum
bit-complexity of
$\softoh(LT(\log q + \log D))$ bit operations. While this paper gets
closer, some considerable improvements remain to be found.

\section{Acknowledgements}

The first two authors would like to acknowledge the support of the Natural Sciences and Engineering Research Council of Canada (NSERC).
The third author is supported by the National Science Foundation,
award \#1319994.

\bibliographystyle{plainnat}

%\bibliography

\end{document}